\documentclass[10pt]{article}
\usepackage{amsmath}
\usepackage{amssymb, amscd, amsthm}
\usepackage[all]{xy}
\usepackage[dvips]{graphicx}
\usepackage{verbatim}
\usepackage[perpage,symbol*]{footmisc}
\newtheorem{theorem}{Theorem}

\newtheorem{lemma}{\textbf{Lemma}}
\newtheorem{remark}{\textbf{Remark}}
\newtheorem{corollary}{\textbf{Corollary}}

\newtheorem{proposition}{\textbf{proposition}}

\newtheorem{Conjecture}{\textbf{Conjecture}}
\usepackage[justification=centering]{caption}
\usepackage{longtable}

\usepackage{multirow}

\newcommand{\F}{\mathbb{F}}

\begin{document}

\baselineskip 17pt
\title{\Large\bf How to Expand a Self-orthogonal Code }

\author{Jon-Lark Kim\qquad\qquad Hongwei Liu\qquad\qquad Jinquan Luo}
\footnotetext{Jon-Lark Kim is with  Dept of Math, Sogang University, Seoul Korea. Hongwei Liu and Jinquan Luo are with School of Mathematics
and Statistics \& Hubei Key Laboratory of Mathematical Sciences, Central China Normal University, Wuhan China 430079.\\
 E-mail: jlkim@sogang.ac.kr(J-K.Kim), hwliu@ccnu.edu.cn(H.Liu), luojinquan@ccnu.edu.cn(J.Luo)}

\date{}
\maketitle
\begin{abstract}
  In this paper, we show how to expand  Euclidean/Hermitian self-orthogonal code preserving their orthogonal property.   Our results show that every $k$-dimension Hermitian self-orthogonal code is contained in a $(k+1)$-dimensional Hermitian self-orthogonal code. Also, for $k< n/2-1$,  every $[n,k]$ Euclidean self-orthogonal code is contained in an $[n,k+1]$ Euclidean self-orthogonal code. Moreover, for $k=n/2-1$ and $p=2$,  we can also fulfill the expanding process. But for $k=n/2-1$ and $p$ odd prime, the expanding process can be fulfilled if and only if an extra condition must be satisfied.  We also propose two feasible algorithms on these expanding procedures.
\end{abstract}

\section{Introduction}
 Let $\mathbb{F}_q$ denote the finite field with $q$ elements. A subset $C$ of $\mathbb{F}_q^n$ is called a $q$-ary linear code. Its parameters are $[n,k,d]$ where $k$ is the dimension and $d$ is  minimum  (Hamming) distance of $C$, which is the minimal number of positions for any two different vectors of $C$ (called codewords) taking different values. The Euclidean dual $C^\perp$ of $C$ is defined as
\[C^\perp=\left\{(c_1^\prime,c_2^\prime,\ldots,c_n^\prime)\in \mathbb{F}_q^n\mid\sum_{i=1}^nc_i'c_i=0,\:\forall\:(c_1,c_2,\ldots,c_n)\in C\right\}.\]
The code $C$ is Euclidean self-orthogonal if it satisfies $C\subseteq C^\perp$. It is Euclidean self-dual if  $C=C^\perp$.

Similarly,  the Hermitian dual $C^{\perp_H}$ of $C$ is defined as
\[C^{\perp_H}=\left\{(c_1^\prime,c_2^\prime,\ldots,c_n^\prime)\in \mathbb{F}_{q^2}^n\mid\sum_{i=1}^n{c_i^\prime}^qc_i=0,\:\forall\:(c_1,c_2,\ldots,c_n)\in C\right\}.\]
The code $C$ is  Hermitian self-orthogonal if  $C\subseteq C^{\perp_H}$ satisfies. It is  Hermitian self-dual if $C=C^{\perp_H}$.

Self-orthogonal codes have both special structures and many applications in different areas. For instance, both Euclidean(\cite{CRSS}) and Hermitian(\cite{KKKS}) self-orthogonal codes can be applied to construct quantum codes. Furthermore, self-dual code, as a special type of self-orthogonal code, can be applied to construct some $t$-designs, such
as \cite{AM} and \cite{GB}. Besides, they are closely related to lattices and modular forms \cite{BBH}, \cite{H}, \cite{Sv}.  Also, self-dual codes have applications in linear secret sharing. In \cite{CDGU}, linear secret sharing schemes with specific access structure are constructed from self-dual codes.

In this paper, we will expand Euclidean and Hermitian self-orthogonal codes. Our methods are based on Gram-Schimidt orthogonalization procedure.

\section{Expanding Hermitian Self-orthogonal Code}
Some notations:
\begin{itemize}
  \item $r=q^2$, $q$ is prime power.
  \item For $\alpha=(a_1, \cdots, a_n),\beta=(b_1, \cdots, b_n)\in \F_r^n$, the Hermitian inner product $\langle \alpha, \beta\rangle_h=a_1^q b_1+\cdots+a_n^q b_n$.
  \item For a linear code $C$ over $\F_r$, we denote by $C^{\perp_h}$ the Hermitian dual code of $C$.
\end{itemize}
\begin{theorem}
For an $[n,k]_r$ Hermtian self-orthogonal code $C$ with $n>2k+1$, there exists an $[n, k+1]_r$ Hermtian self-orthogonal code $C'$ such that $C$ is a subcode of $C'$.
\end{theorem}
\begin{proof}
Since $\dim C=k$, then $\dim C^{\perp_h}=n-k\geq k+2$. Choose a basis of $C$ as $\alpha_1, \cdots, \alpha_k$, which can be extended to a basis of $C^{\perp_h}$ as $\alpha_1, \cdots, \alpha_k, \beta_1, \cdots, \beta_{n-2k}$. Note that $n-2k\geq 2$.
\begin{itemize}
  \item If there exists some $i$ satisfying $\langle \beta_i, \beta_i\rangle_h=0$, then $C'$ linearly generated by $\alpha_1, \cdots, \alpha_k, \beta_i$ is  a Hermtian self-orthogonal code containing $C$.
  \item Otherwise for any $i$, $\langle \beta_i, \beta_i\rangle_h\neq 0$. Then Gram-Schmidt orthogonalization  procedure of $\beta_1, \beta_2$ shows that
  \[\gamma_1=\beta_1, \gamma_2=\beta_2-\frac{\langle \beta_1, \beta_2\rangle_h}{\langle \beta_1, \beta_1\rangle_h}\beta_1.\]
  Then $\alpha_1, \cdots, \alpha_k, \gamma_1, \gamma_2$ is linear independent and  $\alpha_1, \cdots, \alpha_k$ is Hermitian orthogonal to $\gamma_1, \gamma_2$. Moreover, $\gamma_1$ is Hermitian orthogonal to $\gamma_2$.  We try to find vector of the form $t\gamma_1+\gamma_2$ with $t\in \F_r$ which is self Hermitian orthogonal. Equation
  \[\langle t\gamma_1+ \gamma_2, t\gamma_1+\gamma_2\rangle_h=0\]
  can be transfomed to
  \[\langle \gamma_1, \gamma_1\rangle_h t^{q+1}+\langle\gamma_2, \gamma_2\rangle_h =0.\]
  Since $\langle \gamma_1, \gamma_1\rangle_h\in \F_q^*$ and $\langle\gamma_2, \gamma_2\rangle_h\in \F_q$,
  \[t^{q+1}=-\frac{\langle\gamma_2, \gamma_2\rangle_h}{\langle \gamma_1, \gamma_1\rangle_h}\]
  has solution(s) in $\F_r$. Precisely,
     \begin{itemize}
       \item if $\langle\gamma_2, \gamma_2\rangle_h=0$, then we only have one solution $t=0$.
       \item if $\langle\gamma_2, \gamma_2\rangle_h\neq 0$, then we have $q+1$ solutions in $\F_r$.
     \end{itemize}
  Nevertheless, let $t_0$ be a solution. Denote by $\gamma=t_0\gamma_1+ \gamma_2$. Then $C'$ linear spanned by $\alpha_1, \cdots, \alpha_k, \gamma$ is  a Hermtian self-orthogonal code containing $C$.
\end{itemize}
\end{proof}

\begin{corollary}
For an $[n,k]_q$ Hermitian self-orthogonal code $C$ with $n=2k+l, l\geq 2$, there exists a tower of Hermitian self-orthogonal linear codes $C_i$ ($1\leq i\leq \lfloor \frac{l}{2}\rfloor$) such that $\dim C_{i+1}=\dim C_i+1$:
\[C=C_0\subseteq C_1\subseteq\cdots\subseteq C_{\lfloor \frac{l}{2}\rfloor}.\]
\end{corollary}
\begin{proof}
Repeating the process in the above theorem, we obtain the code tower.
\end{proof}

For a Hermtitian self-orthogonal $[n,k]$ code $C$ with $n>2k+1$. The expanding algorithm could be depicted as follows.
\begin{itemize}
  \item[Step 1]: Choose a  basis $\Omega=\{\alpha_1, \cdots, \alpha_k\}$ of $C$.
  \item[Step 2]: Choose linear independent $\beta_1, \beta_2\in C^{\perp_h}\backslash C$.
  \item[Step 3]: If some $\beta_i$ is self-orthogonal, set
  \[C_1=\mathrm{span}(\alpha_1, \cdots, \alpha_k, \beta_i).\]
 Otherwise, set
 \[\gamma_1=\beta_1, \gamma_2=\beta_2-\frac{\langle \beta_1, \beta_2\rangle_h}{\langle \beta_1, \beta_1\rangle_h}\beta_1.\]
 Choose $t_0\in \F_r$ satisfying
 \[t_0^{q+1}=-\frac{\langle\gamma_2, \gamma_2\rangle_h}{\langle \gamma_1, \gamma_1\rangle_h}.\]
 Set
 \[C_1=\mathrm{span}(\alpha_1, \cdots, \alpha_k, t_0\gamma_1+ \gamma_2).\]
\end{itemize}
As a result, the $[n,k+1]$ code $C_1$ containing $C$ is also Hermitian self-orthogonal.

In the end of this section we will discuss minimal distance.
\begin{proposition}
Let $C$ be an $[n,k]$ Hermtitian self-orthogonal code over $\F_{q^2}$ with $n>2k+1$. Suppose the dual code of $C$ has minimal distance $d$.   Then for every $[n,k+1]$ Hermtitian self-orthogonal code $C_1$ containing $C$, both $C_1$ and  $C_1^{\perp_h}$ has minimal distance at least $d$.
\end{proposition}
\begin{proof}
The inclusion chain
\[C\subset C_1\subset C_1^{\perp_h} \subset C^{\perp_h}\]
implies
\[\mathrm{d}(C)\geq \mathrm{d}(C_1)\geq \mathrm{d}(C_1^{\perp_h})\geq \mathrm{d}(C^{\perp_h})=d.\]
\end{proof}
Here we present an open problem.
\[\text{\bf Open\; Problem}: \text{fix}\; C\, \text{as\, above, \, how\, can\, we\, find}\; C_1\, \text{with\, minimal\, distance\, as\, large\, as\, possible?}\]

\section{Expanding Euclidean Self-orthogonal Code}
Some notations:
\begin{itemize}
  \item Let $q=p^m$ be prime power.
  \item For $\alpha=(a_1, \cdots, a_n),\beta=(b_1, \cdots, b_n)\in \F_q^n$, the Euclidean inner product $\langle \alpha, \beta\rangle=a_1 b_1+\cdots+a_n b_n$.
  \item For a linear code $C$ over $\F_q$, we denote by $C^{\perp}$ the dual code of $C$.
\end{itemize}

Now we start to expand Euclidean self-orthogonal code.
\begin{theorem}\label{expand2}
For an $[n,k]_q$ Euclidean self-orthogonal code $C$ with $n\geq 2k+3$ , there exists an $[n, k+1]_q$ Euclidean self-orthogonal code $C'$ such that $C$ is a subcode of $C'$ if
\begin{itemize}
  \item[(i)] $p$ odd and $n\geq 2k+3$ or
  \item[(ii)] $p=2$ and  $n\geq 2k+2$.
\end{itemize}
\end{theorem}
\begin{proof}
Since $\dim C=k$, then $\dim C^{\perp}=n-k$. Choose a basis of $C$ as $\alpha_1, \cdots, \alpha_k$, which can be extended to a basis of $C^{\perp}$ as $\alpha_1, \cdots, \alpha_k, \beta_1, \cdots, \beta_{n-2k}$.
\begin{itemize}
  \item If there exists some $i$ satisfying $\langle \beta_i, \beta_i\rangle=0$, then $C'$ linearly generated by $\alpha_1, \cdots, \alpha_k, \beta_i$ is  a Euclidean self-orthogonal code containing $C$.
  \item Otherwise for any $i$, $\langle \beta_i, \beta_i\rangle\neq 0$. Then Gram-Schmidt  orthogonalization  procedure of $\beta_1, \beta_2,\beta_3$ shows that
   \begin{equation}\label{orth}
   \gamma_1=\beta_1, \gamma_2=\beta_2-\frac{\langle \beta_1, \beta_2\rangle}{\langle \beta_1, \beta_1\rangle}\beta_1.
    \end{equation}
    If $\langle \gamma_2, \gamma_2\rangle=0$, then $C'$ linearly generated by $\alpha_1, \cdots, \alpha_k, \gamma_2$ is  a Euclidean self-orthogonal code containing $C$.
     Now we assume $\langle \gamma_2, \gamma_2\rangle\neq 0$.

\begin{itemize}
  \item[(i)]  In the case of $n-2k\geq 3$ and $p$ odd prime.  Set \[\gamma_3=\beta_3-\frac{\langle \gamma_1, \beta_3\rangle}{\langle \gamma_1, \gamma_1\rangle}\gamma_1-\frac{\langle \gamma_2, \beta_3\rangle}{\langle \gamma_2, \gamma_2\rangle}\gamma_2.\]
        Then $\alpha_1, \cdots, \alpha_k, \gamma_1, \gamma_2, \gamma_3$ is linear independent and  $\alpha_1, \cdots, \alpha_k$ is orthogonal to $\gamma_1, \gamma_2,\gamma_3$. Moreover, $\gamma_1, \gamma_2, \gamma_3 $ are pairwise orthogonal.   We try to find vector of the form $\gamma_1+s\gamma_2+t\gamma_3$ with $s,t\in \F_q$ which is Euclidean self-orthogonal. If $\langle \gamma_3, \gamma_3\rangle=0$, then $C'$ spanned by $\alpha_1, \cdots, \alpha_k, \gamma_3$ is  a desired code. In the following we assume $\langle \gamma_3, \gamma_3\rangle\neq 0$. Equation
  \[\langle \gamma_1+s\gamma_2+t\gamma_3,\gamma_1+s\gamma_2+t\gamma_3\rangle=0\]
  can be expanded to
  \[\langle \gamma_1, \gamma_1\rangle +s^2\langle\gamma_2, \gamma_2\rangle+t^2\langle\gamma_3, \gamma_3\rangle =0.\]
  The above equation has approximately $q$ solutions $(s,t)\in \F_q^2$ using quadratic Jacobi sums. However, we could deal it in an elementary way to seek one solution. Consider
  \[S=\left\{\langle \gamma_1, \gamma_1\rangle +s^2\langle\gamma_2, \gamma_2\rangle\,|\, s\in \F_q\right\}, \qquad T=\left\{-t^2\langle\gamma_3, \gamma_3\rangle\,|\, t\in \F_q\right\}.\]
  Since $|S|=|T|=\frac{q+1}{2}$, then $|S\cap T|\geq 1$ which implies that there exists $s_0, t_0\in \F_q$ such that $\gamma_0:=\gamma_1+s_0\gamma_2+t_0\gamma_3$ is Euclidean self-orthogonal. Moreover, $(s_0, t_0)\neq (0,0)$. In this case, $C'$ spanned by $\alpha_1, \cdots, \alpha_k, \gamma_0$ is  a desired code.
  \item[(ii)] In the case of $n-2k\geq 2$ and $p=2$. The equation $\langle\gamma_1+s\gamma_2, \gamma_1+s\gamma_2\rangle$ is transformed into
  \[\langle\gamma_1, \gamma_1\rangle+s^2 \langle\gamma_2, \gamma_2\rangle=0\]
  which has some solution $t=t_0\in \F_{2^{m}}$. Set $\gamma_0=\gamma_1+s_0\gamma_2$. Then the code $C'$ spanned by $\alpha_1, \cdots, \alpha_k, \gamma_0$ is  a desired code.
\end{itemize}

\end{itemize}
\end{proof}

There is one remaining case, that is,  $n=2k+2$ and $p$ is odd prime. The code $C$ in Theorem \ref{expand2} can be expanded to  $C'$ if and only if $-\langle\gamma_1, \gamma_1\rangle \langle\gamma_2, \gamma_2\rangle$ is a square in $\F_q$.

As a result, a tower consisting Euclidean self-orthogonal codes could be depicted.

\begin{corollary}\label{coro2}
For an $[n,k]_q$ Euclidean self-orthogonal code $C$ with $n\geq 2k+2$,
there exists a tower of Euclidean self-orthogonal linear codes $C_i$ ($1\leq i\leq r$) such that $\dim C_{i+1}=\dim C_i+1$:
\[C=C_0\subseteq C_1\subseteq\cdots\subseteq C_{r}.\]
In particular, $C_r$ has dimension
\begin{itemize}
  \item[(1)] $\frac{n-1}{2}$, if $n$ is odd (here $r=\frac{n-1}{2}-k$);
  \item[(2)] $\frac{n}{2}-1$, if $p$ is odd and $n$ is even (here $r=\frac{n}{2}-k-1$);
  \item[(3)] $\frac{n}{2}$, if $p=2$ and $n$ is even (here $r=\frac{n}{2}-k$)
\end{itemize}
\end{corollary}
\begin{proof}
Repeating the process in Theorem \ref{expand2}, we obtain the code tower.
\end{proof}

It should be noted that when $n$ is even and $p$ is an odd prime, the code tower in the above result may not be the longest. Precisely, for $C_{r}$ with $r=\frac{n}{2}-k-1$ in the above code chain, if  linear independent orthogonal pair $\gamma_1, \gamma_2\in C^{\perp}\backslash C$ satisfies $-\langle\gamma_1, \gamma_1\rangle \langle\gamma_2, \gamma_2\rangle$ is a square in $\F_q$, then there exists Euclidean self-dual code $C_{r+1}\supset C_r$.

\begin{remark}
Based on the above results, we may try to find self-orthogonal codes of maximal dimension. Here are some observations.
\begin{itemize}
  \item For any self-orthogonal $[2k+2, k]_q$ code $C$  with $q\equiv 3\pmod 4$ and $k\equiv 0\pmod 2$, there does not exist  self-dual $[2k+2, k+1]_q$ code containing $C$. Assume, on the contrary, $C'$ is a self dual $[2k+2, k+1]_q$ code with systematic generator matrix ($I_{k+1}$ is the identity matrix of rank $k+1$)
      \[G=[I_{k+1}, P].\]
      The self-dual property of $C'$ yields $GG^T=O$ which is equivalent to $I_{k+1}=-PP^T$. Taking determinants on both sides, we obtain $-1=\left(\det P\right)^2$ which is absurd since $-1$ is not a square in $\F_q^*$.
  \item  For a $[4,1]$ Euclidean self-orthogonal linear code $C$ generated by $(a_1, a_2, a_3, a_4)$. We may assume $a_1a_2\neq 0$. Then the code $C$ can be expanded to a self-dual code $C'$ generated by $(a_1, a_2, a_3, a_4)$ and $(-a_3, -a_4, a_1, a_2)$.
\end{itemize}
\end{remark}

For a Euclidean self-orthogonal $[n,k]$ code $C$ with $n>2k+2$. The expanding algorithm could be depicted as follows.
\begin{itemize}
  \item[Step 1]: Choose a  basis $\Omega=\{\alpha_1, \cdots, \alpha_k\}$ of $C$.
  \item[Step 2]: Choose linear independent $\beta_1, \beta_2,\beta_3\in C^{\perp_h}\backslash C$.
  \item[Step 3]:
    \begin{itemize}
      \item[Case 1]  If $\beta_1$ is self-orthogonal, set
       \[C_1=\mathrm{span}(\alpha_1, \cdots, \alpha_k, \beta_1).\]
      \item[Case 2] If $\beta_1$ is not self-orthogonal,
            set
          \[\gamma_1=\beta_1, \gamma_2=\beta_2-\frac{\langle \beta_1, \beta_2\rangle}{\langle \beta_1, \beta_1\rangle}\beta_1.\]
            \begin{itemize}
              \item[Subcase 1] If $\gamma_2$ is self-orthogonal, set
       \[C_1=\mathrm{span}(\alpha_1, \cdots, \alpha_k, \gamma_2).\]
              \item[Subcase 2] If $\gamma_2$ is not self-orthogonal, calculate
              \[\gamma_3=\beta_3-\frac{\langle \gamma_1, \beta_3\rangle}{\langle \gamma_1, \gamma_1\rangle}\gamma_1-\frac{\langle \gamma_2, \beta_3\rangle}{\langle \gamma_2, \gamma_2\rangle}\gamma_2.\]
              Find $s,t\in \F_q$ satisfies
              \[\langle \gamma_1, \gamma_1\rangle +s^2\langle\gamma_2, \gamma_2\rangle+t^2\langle\gamma_3, \gamma_3\rangle =0.\]
              Set
               \[C_1=\mathrm{span}(\alpha_1, \cdots, \alpha_k, \gamma_1+ s\gamma_2+t\gamma_3).\]
            \end{itemize}

    \end{itemize}
\end{itemize}
As a result, the $[n,k+1]$ code $C_1$ containing $C$ is also Euclidean self-orthogonal. We can also present an algorithm for $n=2k+2$ if $-\langle\gamma_1, \gamma_1\rangle \langle\gamma_2, \gamma_2\rangle$ is a square in $\F_q$ in a similar way.

It should be noted that in most cases, Euclidean self-dual codes exist, which has been shown in \cite{BCR}, Prop. 3.6; \cite{GR}, Theo. 7;  \cite{M}, pp. 150 and \cite{P}, Cor. 3.1. Here we provide a feasible algorithm to increase the dimension of Euclidean self-orthogonal code by one at each procedure until we get a Euclidean self-dual code(if exists).

In the end of this section we discuss minimal distance. The result and proof are similar to the Hermitian case.
\begin{proposition}
Let $C$ be an $[n,k]$ Euclidean self-orthogonal code over $\F_{q}$ with $n>2k+2$ whose dual code has minimum distance $d$.  Then for every $[n,k+1]$ Euclidean self-orthogonal code $C_1$ containing $C$, both $C_1$ and  $C_1^{\perp_h}$ has minimal distance at least $d$.
\end{proposition}
\begin{proof}
The inclusion chain
\[C\subset C_1\subset C_1^{\perp} \subset C^{\perp}\]
implies
\[\mathrm{d}(C)\geq \mathrm{d}(C_1)\geq \mathrm{d}(C_1^{\perp})\geq \mathrm{d}(C^{\perp})\geq d.\]
\end{proof}

Here we can  propose an open problem.

{\bf Open Problem:} For a fixed Euclidean self-orthogonal code $C$ with parameters $[n,k]$, how can we find   an $[n,k+1]$ Euclidean self-orthogonal code $C_1\supset C$ with largest minimal distance?

\section{Conclusion}
\quad In this paper we discussed the expanding process of Hermitian/Euclidean self-orthogonal codes. Our result showed that in the Hermitian case, every  Hermitian self-orthogonal code can be expanded to a larger Hermitian self-orthogonal code, in which every step the dimension increases by 1. This process terminates at Hermitian self-dual code(even length) or almost self-dual code(odd length).   While in Euclidean case, this process terminates at self-dual code if one more condition should be satisfied. At the end of each case, we discussed the minimal distance of the expanded code $C_1$. Meanwhile, we proposed two open problems on how to find $C_1$ with minimal distance as large as possible.


\begin{thebibliography}{1}
\bibitem{AM} E. F. Assmus, Jr. and H. F. Mattson, Jr,
        New 5-designs, J. Comb. Theory, vol. 6, no. 2, pp. 122-151, Mar. 1969.
 \bibitem{BBH} S. Bouyuklieva, I. Bouyukliev, and M. Harada,
        Some extremal self-dual codes and nunimodular lattices in dimension $40$, Finite Fields Appl., vol. 21, pp. 67-83, May 2013.
 \bibitem{BCR} E. Berardini, A. Caminata and A. Ravagnani, Structure of CSS and CSS-T quantum codes, Designs, Codes and Crypt., vol. 92, no. 10
pp. 2801-2823, 2024.
\bibitem{CDGU} R. Cramer, V. Daza, I. Gracia, J. J. Urroz,  C. Padr$\mathrm{\acute{o}}$, On codes, matroids and secure multi-party computation from linear secret sharing schemes, Adv. Cryptology-CRYPTO2005 (Lecture Notes in Computer Science). Berlin, Germany: Springer-Verlag, 2005, vol. 3621, pp. 327-343.
\bibitem{CRSS} A. R. Calderbank, E. M. Rains, P. W. Shor, and N. J. A. Sloane,
        Quantum error correction via codes over $GF(4)$, IEEE Trans. Inform. Theory,  vol. 44, no. 4, pp. 1369-1387, Jul. 1998.
 \bibitem{GB} C. Bachoc and P. Gaborit,
        Designs and self-dual codes with long shadows, J. Combin. Theory Ser. A,  vol. 105, no. 1, pp. 15-34, Jan. 2004.
 \bibitem{GR} H. Gluesing-Luerssen, and A. Ravagnani,   $l$-Complementary Subspaces and Codes in Finite Bilinear
Spaces. IEEE Trans. on Inform. Theory, vol. 70, vol. 4,  pp. 2443-2455, April 2024.
 \bibitem{H} M. Harada, On the existence of frames of the Niemeier lattices and self-dual codes over $\mathbb{F}_p$,  J.Algebra, vol. 321, no. 8, pp. 2345-2352, April 2009.
 \bibitem{M} E. Moorhouse, Incidence geometry,  Univ. Wyoming, Laramie,
WY, USA, Lecture Notes, 2017. [Online]. Available: https://
ericmoorhouse.org/handouts/Incidence$\underline{~}$Geometry.pdf
\bibitem{PLE} V.Pless, On the uniqueness of the Golay codes, J. of Comb. Theory, Series A,  vol. 5, no. 3, pp. 215-228, 1968.
\bibitem{P} V. Pless,  On Witt's theorem for nonalternating symmetric bilinear forms
over a field of characteristic 2,  Proc. Amer. Math. Soc., vol. 15, no. 6, pp. 979-983, 1964.

\bibitem{KKKS}A. Ketkar, A. Klappenecker, S. Kumar and P.K. Sarvepalli, Nonbinary stabilizer codes over finite fields. IEEE Trans. on Inform. Theory, vol.52, no. 11, pp. 4892-4914, Nov.2006.

 \bibitem{Sv} F.-W. Sun and H. C. A. van Tilborg,
        The Leech lattice, the octacode, and decoding algorithms, IEEE Trans. Inform. Theory,  vol. 41, no. 4, pp. 1097-1106, July 1995.
\end{thebibliography}
\end{document}